\documentclass[12pt]{amsart}

\usepackage{amsthm,cite,color,graphicx}
\usepackage{pstricks}

\newcommand{\cP}{\mathcal{P}}
\newcommand{\cZ}{\mathcal{Z}}
\newcommand{\cE}{\mathcal{E}}
\newcommand{\cF}{\mathcal{F}}
\newcommand{\cB}{\mathcal{B}}

\newcommand{\R}{\mathbb{R}}

\newtheorem{theorem}{Theorem}
\newtheorem{lemma}[theorem]{Lemma}
\newtheorem{proposition}[theorem]{Proposition}

\theoremstyle{definition}
\newtheorem{remark}[theorem]{Remark}
\newtheorem{definition}[theorem]{Definition}

\newcommand\term[1]{{\em #1\/}}
\newcommand\bterm[1]{{\bf #1\/}}

\begin{document}

\title[Nodal deficiency of billiard eigenfunctions]{Critical partitions and nodal deficiency of billiard eigenfunctions}

\author{Gregory Berkolaiko$^{1}$, Peter Kuchment$^{1}$, and Uzy Smilansky$^{2,3}$}
\address{$^{1}$ Department of Mathematics, Texas A\&M University, College Station, TX 77843-3368 USA.}
\address{$^{2}$ Department of Physics of Complex Systems,
Weizmann Institute of Science, Rehovot 76100, Israel.}
\address{$^{3}$ School of Mathematics, Cardiff University, Cardiff,
Wales, UK}

\date{}
\maketitle

\begin{abstract}
  The paper addresses the nodal count (i.e., the number of nodal
  domains) for eigenfunctions of Schr\"{o}dinger operators with
  Dirichlet boundary conditions in bounded domains. The classical
  Sturm theorem states that in dimension one, the nodal and
  eigenfunction counts coincide: the $n$-th eigenfunction partitions
  the interval into $n$ nodal domains. The Courant Nodal Theorem
  claims that in any dimension, the number of nodal domains $\nu_n$ of
  the $n$th eigenfunction cannot exceed $n$. However, it follows from
  an asymptotically stronger upper bound by Pleijel that in dimensions
  higher than $1$ the equality can hold for only finitely many
  eigenfunctions. Thus, in most cases a ``nodal deficiency''
  $d_n=n-\nu_n$ arises.  One can say that the nature of the nodal
  deficiency has not been understood.

  It was suggested in recent years that, rather than starting with
  eigenfunctions, one can look at partitions of the domain into
  $\nu$ sub-domains, asking which partitions can correspond to
  eigenfunctions, and what would be the corresponding deficiency.  To
  this end one defines an ``energy'' of a partition, for example, the
  maximum of the ground state energies of the sub-domains.  One notices
  that if a partition does correspond to an eigenfunction, then the
  ground state energies of all the nodal domains are the same, i.e.,
  it is an equipartition.  It was shown in a recent paper by Helffer,
  Hoffmann-Ostenhof and Terracini that (under some natural conditions)
  partitions minimizing the energy functional correspond to the
  ``Courant sharp'' eigenfunctions, i.e. to those with zero nodal
  deficiency.

  In this paper it is shown that it is beneficial to restrict the
  domain of the functional to the equipartition, where it becomes
  smooth.  Then, under some genericity conditions, the nodal
  partitions correspond exactly to the critical points of the
  functional.  Moreover, the nodal deficiency turns out to
  be equal to the Morse index at the corresponding critical
  point. This explains, in particular, why the minimal partitions
  must be Courant sharp.
\end{abstract}

\section{Introduction}

We consider the Schr\"{o}dinger operator
\begin{equation}
 H=-\Delta + V(x)
\end{equation}
with Dirichlet conditions in a connected bounded domain $\Omega\subset\R^d$. We
will assume that the domain has a smooth boundary and that the real
potential $V$ is also smooth on $\overline{\Omega}$.  While these assumptions are overly
restrictive, we do not want to burden our considerations with less
significant details. See the final section for additional remarks.

The operator $H$ can be defined via its quadratic form
$$
h[u,u]=\int\limits_\Omega |\nabla u(x)|^2 dx + \int\limits_\Omega V(x)|u(x)|^2 dx
$$
with the domain $H^1_0 (\Omega)$. Thus defined, it is self-adjoint in $L_2(\Omega)$ and has real discrete spectrum of finite multiplicity
$$
\lambda_1<\lambda_2\leq \lambda_3\leq \dots,
$$
where $\lim\limits_{n\to\infty}\lambda_n=\infty.$ It has an orthonormal basis of real-valued eigenfunctions $\psi_n$ such that $\psi_1(x)> 0$. We will sometimes use the notations $H(\Omega)$ and $\lambda_j(\Omega)$, when we need to emphasize the dependence of the operator and its spectrum on the domain.

\begin{definition}
For a function $f(x)$ we will be interested in its \textbf{nodal (zero) set}
$$
\cZ(f):=f^{-1}(0)=\overline{\{x\in\Omega\mid f(x)=0\}}.
$$
The complement $\Omega\setminus N(f)$ is the union of connected open sub-domains $D_1,\dots,D_\nu$ of $\Omega$, which we will call \textbf{nodal domains}. The nodal domains form the \textbf{nodal partition}  $P(f)=\{D_j\}$ corresponding to the function.
\end{definition}

We will be mostly interested in the case when $f(x)$ is an
eigenfunction $\psi_n$, and thus in its nodal set $\cZ(\psi_n)$, its
nodal domains $D_1,\dots,D_\nu$ of $\Omega$, and its \term{nodal
  partition} $P(\psi_n)=\{D_j\}$.

A lot of attention has been paid to the nodal structure of eigenfunctions (e.g., \cite{GoodVibr,Hal_McL96,NadTotYak_umn01,DonFef_ancet90,DonFef_apde90,DonFef_jams90,DonFef_jga92,DonFef_inv88} and references therein), and in particular to the number $\nu$ (or $\nu_{\psi_n}$, if one wants to emphasize dependence on the eigenfunction) of nodal domains of the $n$-th eigenfunction $\psi_n$ of $H$.

In spite of more than 300 years history of this topic\footnote{Robert
  Hooke observed on 8 July 1680 the nodal patterns on vibrating glass
  plates, running a bow along the edge of a glass plate covered with
  flour \cite{Hooke}. A hundred years later, the same effect was
  systematically studied by E. Chladni \cite{Chladni}. In fact, such
  patterns were known to Leonardo da Vinci \cite{davinci} and Galileo Galilei \cite{Galileo}. See also some
  historical discussion in \cite{GoodVibr}.}, open questions still
abound. We will discuss here one of them, the issue of the so called
nodal deficiency. The classical Sturm theorem states that in
dimension one, the nodal and eigenfunction counts coincide:
$\nu_{\psi_n}=n$.  The Courant Nodal Theorem \cite[Vol. I, Sec. V.5,
VI.6]{CouHil_book53} asserts that in any dimension, the upper bound on
the number of nodal domains still holds:
$$
\nu_{\psi_n}\leq n.
$$
While $\nu_{\psi_1}=1$ and $\nu_{\psi_2}=2$, it follows from an
asymptotically stronger upper bound for $\nu_{\psi_n}$ by Pleijel
\cite{Ple_cpa56} that in higher dimensions the equality
$\nu_{\psi_n}=n$ can hold only for a finitely many values of $n$.
Moreover, there are known example of eigenfunctions $\psi_n$ with
arbitrarily large index $n$ that have just two nodal domains.

The eigenfunctions $\psi_n$ for which $\nu_{\psi_n}=n$, are sometimes
called \term{Courant sharp}. The non-negative difference
$$
d_n:=n-\nu_{\psi_n}
$$
is said to be the \textbf{nodal deficiency} of an eigenfunction.  It is
believed that the integer sequence $\{d_n\}$ contains much information
about the geometry of the domain (see, for example,
\cite{BluGnuSmi_prl02,KarSmi_jpa08,Kla_jpa09,GnuSmiSon_jpa05,BruKlaPuh_jpa07}).
However, we are not aware of any interpretation of the meaning of
individual nodal deficiencies $d_n$.  It is the goal of this text to
present one such interpretation.

An important new approach has been developed in the last several years
in the series of papers
\cite{ConTerVer_cvp05,Hel_sem07,HelHofTer_aip09}. Namely, instead of
concentrating on an eigenfunction, one can look at a \term{partition}
of $\Omega$ by connected open domains $\{D_j\}_{j=1}^\nu$ and try to
determine whether a given partition can be the nodal partition of an
eigenfunction of a given operator $H$, and if yes, what could be the
corresponding nodal deficiency.

Some necessary conditions on the partition are not hard to find. Indeed,
one can introduce the \textbf{graph of the partition}, with each
partition domain $D_j$ serving as a vertex and any two nodal domains
that have a $(d-1)$-dimensional common boundary being
connected by an edge.  Since, due to standard uniqueness theorems for
the Cauchy problem, the eigenfunction must change its sign when
crossing any $(d-1)$-dimensional piece of the boundary of two adjacent
nodal domains, we conclude that the following well known property
holds:

\begin{proposition}
  The graph of the nodal partition corresponding to an eigenfunction
  is bipartite.
\end{proposition}

Another important simple observation is:
\begin{proposition}
  \label{P:ground_states}
  If $\{D_j\}_{j=1}^\nu$ is the nodal partition corresponding to the
  eigenfunction $\psi_n$ with eigenvalue $\lambda_n$, then for
  each nodal domain $D_j$, one has
  $$
  \lambda_1(D_j)=\lambda_n(\Omega).
  $$
  (Here, as before, $\lambda_k(D)$ denotes the $k$th Dirichlet
  eigenvalue in the domain $D$.)
\end{proposition}
Indeed, by definition, $\psi_n$ does not change sign in $D_j$ and thus
is proportional to the groundstate for $H(D_j)$.

This observation leads to the following notion that plays a crucial role in what follows:

\begin{definition}
A partition $P=\{D_j\}_{j=1}^\nu$ is said to be an {\bf equipartition}, if the lowest eigenvalues of all operators $H(D_j)$ are the same, i.e.
\begin{equation}\label{E:equipartition}
   \lambda_1(D_1)=\lambda_1(D_2)=\dots=\lambda_1(D_\nu).
\end{equation}
\end{definition}
As we have already mentioned, every nodal partition (i.e., the partition corresponding to an eigenfunction) is an equipartition.

This observation has lead to the following construction: given a natural number $\nu$, consider the ``space'' of ``arbitrary'' (with some natural restrictions) \textbf{$\nu$-partitions} $\{D_j\}$, i.e. partitions with $\nu$ sub-domains. Let us also introduce the functional
\begin{equation}
  \label{eq:Lambda_def}
  \Lambda(\{D_j\}):=\max\limits_{j=1,\dots,\nu}\lambda_1(D_j)
\end{equation}
on this space.

One can look now at the \textbf{minimal partitions}, i.e.~the partitions
that minimize the functional $\Lambda(\{D_j\})$ for a given $\nu$.
Such minimal partitions are known to exist \cite{ConTerVer_cvp05} and the following important result holds:

\begin{theorem}[Helffer, Hoffmann-Ostenhof, Terracini \cite{HelHofTer_aip09}]
  Minimal bipartite partitions are exactly the nodal partitions of
  Courant sharp eigenfunctions.
\end{theorem}

Our aim is to understand whether there is something that distinguishes
the nodal partitions of the eigenfunctions that are not Courant sharp
(which are the overwhelming majority) and what determines their nodal
deficiencies. It was shown in a recent work \cite{BanBerRazSmi10} that
in the quantum graph situation this question can be answered (see also
\cite{BerRazSmi_jpa12} for the discussion of the discrete graph
case). Inspired by this development, we address here the
eigenfunctions of the operator $H$ defined above. As we will see, one
has to consider not only the minima, but all the critical points of
the functional $\Lambda$ on the ``manifold'' of equipartitions.
Furthermore, it turns out that the Morse index of a critical point
coincides with the nodal deficiency of the corresponding
eigenfunction.

Before we formulate our assumptions we mention that \emph{generically}
with respect to perturbations of the potential $V(x)$ and/or the
domain, the following conditions are satisfied (see
\cite{Uhl_amj76,Alb_psp71,Alb_tams78,NadTotYak_umn01,Mic_amp73,Mic_asn72,IvaKotKre_du77}
and references therein).
\begin{definition}
  \label{def:genericity_eigenfunction}
  We say the eigenfunction $\psi_n$ of $H$ is \bterm{generic} if
  \begin{enumerate}
  \item The corresponding eigenvalue $\lambda_n$ is simple
  \item Zero is a regular value of the eigenfunction $\psi_n$ inside
    $\Omega$ (i.e., $\nabla\psi_n(x)\neq 0$ whenever $\psi_n(x)=
    0$). The normal derivative $\partial\psi_n/\partial \emph{n}$ of
    the eigenfunction $\psi_n$ on the boundary of $\Omega$ has zero as
    its regular value (i.e., the tangential to $\partial\Omega$
    gradient of $\partial\psi_n/\partial \emph{n}$ does not vanish
    whenever $\partial\psi_n/\partial \emph{n} (x)=0$).
  \item The nodal set $N = N(\psi_n)$ is the finite union of
    non-intersecting smooth hyper-surfaces
      $$
      N = \left(\bigcup\limits_k C_k\right) \cup \left(\bigcup\limits_l B_l\right)
      $$
      (see Fig. \ref{F:partition}), where
  \item Each $C_k$ is a closed smooth hyper-surface in $\Omega$.
  \item Each $B_l$ is a smooth hyper-surface,
    whose boundary lies in $\partial\Omega$; $B_l$ intersects
    $\partial\Omega$ transversally.
\end{enumerate}
\end{definition}

In fact, the intersections of $B_l$ with $\partial\Omega$ are
orthogonal, but we will not need to use this information.

\begin{figure}[t]
  \begin{pspicture}(5,4)
    \psccurve[linewidth=3pt,showpoints=false](0.8,1)(0,2)(1.2,3)(4,3.5)(4.5,2)(4,1)(0.8,1)
    \pscurve[linewidth=1pt,showpoints=false](0.8,1)(1.4,2.4)(1.2,3)
    \psccurve[linewidth=1pt,showpoints=false](2.2,1.4)(3,2.7)(4,2.8)(3.2,1.4)
    \rput(2.4,3){\Large $\Omega$}
    \rput(1.3,1.4){\large $B_l$}
    \rput(2.5,1.6){\large $C_k$}
  \end{pspicture}
 \caption{A generic partition}\label{F:partition}
\end{figure}
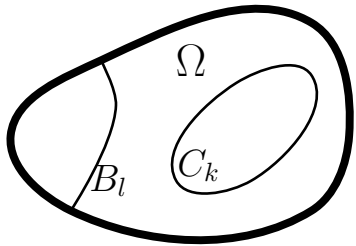

We will be dealing with the {\bf generic situation} only, in the sense
that all conditions (1)-(5) above are satisfied for the eigenfunction
in question.  Partitions that we consider are required to satisfy
similar assumptions.
\begin{definition}
  A partition $P=\{D_j\}$ of $\Omega$ will be called \bterm{generic}, if
  its nodal set $N =  \Omega\setminus\bigcup D_j$ satisfies conditions
  (3)-(5) above.
\end{definition}

We now briefly describe the constructions and results of the paper.

Given a generic $\nu$-partition $P$ and a sufficiently small positive
number $\rho>0$, we will introduce the set $\cP_\rho$ of
$\nu$-partitions that are ``close'' to $P$ in an appropriate
sense. Here $\rho$ indicates a measure of closeness. The set
$\cP_\rho$ will be equipped with the structure of a real Hilbert
manifold by identifying it with a ball in an appropriate functional
Hilbert space.  We will denote by $\cE_\rho$ the subset of $\cP_\rho$
that consists of equipartitions.

On the set $\cE_\rho$ of equipartitions one can consider the functional
$$
\Lambda: \cE\mapsto \R
$$
that maps a partition $P$ into the (common) lowest energy
$\lambda_1(P_j)$ of any sub-domain $P_j$. The notation $\Lambda$ does
not contradict the one used previously in \eqref{eq:Lambda_def}, since
on equipartitions the two functionals obviously coincide. We will also
need some other extensions of the functional $\Lambda$ from the set
$\cE_\rho$ of equipartitions to the whole $\cP_\rho$. Let
$c=(c_1,\dots,c_\nu)\in\R^\nu$ be a unit simplex vector, i.e. such
that $c_j\geq 0$ and $\sum c_j=1$. We define the functional
$\Lambda_c$ on $\cP_\rho$ as follows:
$$
\Lambda_c (P)=\sum c_j\lambda_1(P_j),
$$
where $P\in\cP_\rho$ and $P_j$ are the sub-domains of this partition.
It is obvious that for any unit simplex vector $c$ the restriction of
$\Lambda_c$ to $\cE_\rho$ coincides with $\Lambda$.

We will need the following auxiliary result:
\begin{proposition}\label{P:manifold}\indent
  \begin{enumerate}
  \item For any $c$, the functional $\Lambda_c$ on $\cP_\rho$ is $C^\infty$-smooth.
  \item For a sufficiently small $\rho$, $\cE_\rho$ is a smooth
    sub-manifold of $\cP_\rho$ of co-dimension $\nu-1$.
  \item The functional $\Lambda$ on $\cE_\rho$ is $C^\infty$-smooth.
  \end{enumerate}
\end{proposition}

This allows us now to formulate the first of the two main results of the paper:

\begin{theorem}\label{T:critical}
  Let $P$ be a generic bipartite equipartition of a smooth domain
  $\Omega$. Then, the following statements are equivalent:
  \begin{enumerate}
  \item \label{I:critical:nodal} $P$ is nodal (i.e., $P$ is the nodal
    partition of an eigenfunction $\psi$ of $H$).
  \item \label{I:critical:Lambda_c} There exists a vector
    $c=(c_1,\dots,c_\nu)\in\R^\nu$, $c_j\geq 0$, $\sum c_j=1$, such
    that $P$ is a critical point of the functional $\Lambda_c$ on
    $\cP_\rho$. In this case,
    $$
    c_j=\|\psi\|^2_{L^2(P_j)}=\int\limits_{P_j}|\psi(x)|^2 dx.
    $$
  \item \label{I:ctitical:on_equipartns} $P$ is a critical point of the
    functional $\Lambda$ on $\cE_\rho$.
  \end{enumerate}
  If $\Omega$ is simply connected, the above statements are also
  equivalent to the following,
  \begin{enumerate}
  \item[(4)] At any of the boundary surfaces $C_j$ and $B_l$, the
    normal derivatives of the groundstates for the two adjacent
    sub-domains are proportional.
  \end{enumerate}
\end{theorem}

\begin{remark}
  If the domain $\Omega$ is simply connected, the partition graph of a
  generic partition is a tree and thus bipartite automatically.
  Indeed, if the partition graph is \emph{not} a tree, we can remove
  at least one edge while keeping the graph connected.  That means
  that $\Omega$ will also remain connected after removal of the
  corresponding piece of the partition boundary, $C_j$ or $B_l$ (call
  it $B$, for simplicity). However, the hyper-surface $B$ is either closed
  or has a boundary in $\partial \Omega$. Therefore, due to simply-connectedness
  of $\Omega$, its removal should disconnect $\Omega$.
\end{remark}

Let $P$ be the partition corresponding to a generic eigenfunction.
Theorem~\ref{T:critical} implies that, up to corrections of higher
order, the functional $\Lambda$ is given by
\begin{equation*}
  \Lambda(P + \Delta P) = \Lambda(P) + F_2(\Delta P, \Delta P) +
  \mbox{higher order terms},
\end{equation*}
where $F_2$ --- the Hessian --- is a quadratic form on the tangent
space of $\cE_\rho$ at $P$ (which will be described in detail in
Section \ref{S:manifolds}).  The \textbf{Morse index} of $P$ is then
defined as the maximal dimension of the subspace on which $F_2$ is
negative definite.  Informally, the Morse index counts directions in
which the function $\Lambda$ is unstable, i.e. decreasing in value.
We also define the \textbf{$\mu^0$ - index} as the maximal dimension of a
subspace on which $F_2$ is non-positive.  If the Morse and $\mu^0$
indices coincide, the critical point is called \textbf{non-degenerate}.

Our second main result is the following interpretation of the nodal deficiency:

\begin{theorem}\label{T:morse}
  Let $\psi_n$ be a generic (see Definition~\ref{def:genericity_eigenfunction})
  eigenfunction of $H$ and $P$ be its nodal partition. Then $P$ is a
  non-degenerate critical point of $\Lambda$ restricted to $\cE_\rho$
  and the nodal deficiency $d_n=n-\nu_{\psi_n}$ is equal to the Morse
  index of $\Lambda$ at the point $P$.
%
\end{theorem}

\begin{remark}
  To summarize Theorems~\ref{T:critical} and \ref{T:morse}, instead of
  looking at the minimal points of $\Lambda$, one has to look at the
  critical min-max points, where the maximum is taken over a subspace
  of dimension equal to the nodal deficiency $d_n$. This explains, in
  particular, why the minimal partitions correspond the Courant sharp
  eigenfunctions only.
\end{remark}

The structure of this article is as follows. Section \ref{S:variation} contains a brief exposition of the well known Rayleigh-Hadamard formula for the derivative of an eigenvalue with respect to the domain variation. The manifolds $\cP_\rho$ and $\cE_\rho$ are introduced in Section \ref{S:manifolds}, where also the Proposition \ref{P:manifold} is proven. Theorems \ref{T:critical} and \ref{T:morse} are proven in sections \ref{S:critical} and \ref{S:morse}. Section \ref{S:remarks} contains final remarks and conclusions. In particular, it offers various possible generalizations of the results of this paper.

\section{Domain variation formulas}\label{S:variation}

In this section we provide the formulas for eigenvalue perturbation due to domain variation, which will be important for our considerations. Such formulas have a long history, going back to J.~Rayleigh \cite{Ray_sound} and J.~Hadamard \cite{Had_book08} and are still being developed (see \cite{IvaKotKre_du77,Gar_book86,Gar_Sch_jam53,Gri_jot10,Koz_jde06,Fri_ijm92,Hezari,Peetre,Fuji,FujOza}
for further results and references).

Let $D$ be a proper sub-domain of $\Omega$ (i.e., the closure of $D$
belongs to $\Omega$) with a smooth boundary $C=\partial D$. Later on
in this text $D$ will be one of the sub-domains of a partition. We
denote by $\psi_1(D)$ and $\lambda_1(D)$ the positive
normalized groundstate and the corresponding eigenvalue of $H(D)$.

We are interested in the variation of $\lambda_1(D)$ with respect to
infinitesimal smooth deformations of the boundary $C$. To make it
precise, let us consider the unit external normal vector field on $C$
and extend it into a neighborhood $U$ of $C$ to a smooth unit length
vector field $N(x)$ whose trajectories are the normals to $C$. Let us
now also have a sufficiently smooth real valued function $f(x)$ in
$U$. Consider the normal to $C$ vector field $f(x)N(x)$ and the
corresponding evolution operators $G_t$ of the ``time'' $t$ shift
along the trajectories of this field. They are defined for
sufficiently small values of $t$ and produce deformed surfaces
$C_t=G_t C$ and the variable domains $D_t$ bounded by these surfaces.
Correspondingly, the ground state eigenvalue $\lambda_1(D_t)$ is a
function of $t$. The following result (Rayleigh-Hadamard formula) is
well known
(e.g. \cite{IvaKotKre_du77,Gar_book86,Had_book08,Gar_Sch_jam53,Ray_sound,Gri_jot10,Koz_jde06,Fri_ijm92,Peetre,Fuji,FujOza})
for the case when $V(x)=0$. However, its proof (e.g. the one in
\cite{Gri_jot10}) is valid for non-zero potentials as well.

\begin{theorem}\label{T:eig_var}
The $t$-derivative at $t=0$ of the eigenvalue $\lambda_1(D_t)$ is given by the formula
\begin{equation}\label{E:hadamard}
    \lambda_1^\prime =-\int\limits_C \left(\frac{\partial \psi_1(D)}{\partial n}\right)^2f(x)dS,
\end{equation}
where $\partial/\partial n$ denotes the external normal derivative on
$C$ and $\psi_1(D)$ is, as before, the normalized Dirichlet
ground-state of the domain $D$.
\end{theorem}

\begin{remark}\label{R:eig_var}
  If one uses a smooth vector field $M(x)$ instead of
  the unit normal vector field $N(x)$, an analog of formula
  (\ref{E:hadamard}) for the $t$-derivative at $t=0$ of the eigenvalue
  $\lambda_1(D_t)$ is given by
\begin{equation}\label{E:Mhadamard}
    \lambda_1^\prime =-\int\limits_C \left(\frac{\partial \psi_1(D)}{\partial n}\right)^2f(x)M(x)\cdot N(x)dS,
\end{equation}
where $\partial/\partial n$ denotes the external normal derivative on $C$ and $M\cdot N$ is the inner product of vectors $M$ and $N$.
\end{remark}

Some of the nodal sub-domains will reach the boundary. We thus also need to consider the case of $D\subset \Omega$ that is cut out from $\Omega$ by a smooth surface (curve when $d=2$) $B$ transversal to the boundary $\partial\Omega$ (Fig. \ref{F:cutout}).
\begin{figure}[t]
  \begin{pspicture}(5,4)
    \pscurve[linewidth=1pt,linestyle=dashed](2,3.4)(4,3.5)(4.5,2)(4,1)(2,0.6)
    \pscurve[linewidth=2pt](2,0.6) (0.8,1)(0.1,2)(2,3.4)
    \pscurve[linewidth=1pt,showpoints=false](2,0.6)(2.6,1.2)(2.5,2)(2,3.4)
    \pscurve[linewidth=1pt,showpoints=false,linecolor=red](2.4,0.6)(2.9,1.2)(2.8,2.2)(2.4,3.53)
    \psline[linewidth=1.5pt,linecolor=red]{->}(2,0.6)(2.4,0.6)
    \psline[linewidth=1.5pt,linecolor=red]{->}(2.6,1.2)(2.9,1.2)
    \psline[linewidth=1.5pt,linecolor=red]{->}(2.45,2.1)(2.8,2.2)
    \psline[linewidth=1.5pt,linecolor=red]{->}(2,3.4)(2.4,3.53)
   \rput(3.5,2.5){\Large $\Omega$}
    \rput(2.2,1.4){\large $B$}
    \rput(1.3,1.9){\large $D$}
  \end{pspicture}
  \caption{Domain $D$ cut out from $\Omega$ by $B$ and its deformation.}\label{F:cutout}
\end{figure}
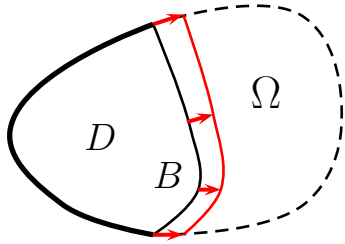

Consider again the external unit normal vector field $N(x)$ to $B$. We
can assume that this field is modified near the boundary
$\partial\Omega$ to a smooth vector field $M(x)$ that is (i)
non-tangential to $B$ and (ii) tangential to $\partial\Omega$ on the
intersection $B\cap\partial\Omega$. It can then be extended to a
smooth vector field $M$ near $B$ such that the trajectories that start
at the boundary points (i.e., points of $B\cap\partial\Omega$) stay on
the boundary $\partial\Omega$ and such that the level sets $G_tB$ for
small values of $t$ are transversal to $\partial\Omega$. Let $f\in
H^s(B)$ (for a sufficiently large $s$) be a real valued function on
$B$. Then, similarly to the case of an internal part of $\Omega$
considered above, one can deform the surface $B$ by a vector field
$fM$, which will define sub-domains $D_t$ of $\Omega$ with their
boundaries $B_t$ transversal to $\partial\Omega$. With this setup, the
variation formula completely identical to (\ref{E:Mhadamard}) holds.

%

Although in this case the variation formula may or may not have been written before, the proof of Theorem \ref{T:eig_var} given, for instance, in \cite{Gri_jot10} carries through without any change.

\begin{remark}\label{R:sign}\indent
If the boundary of the domain $D$ is not connected, then the same variation formula holds, involving the sum of integrals over each connected component $C_j$ and $B_l$ of the boundary.
\end{remark}

\section{Manifolds of partitions. Proof of Proposition \ref{P:manifold}}\label{S:manifolds}

Let us now consider a {\bf generic} $\nu$-partition
$P=\{P_j\}_{j=1}^\nu$ of $\Omega$. We need to introduce a manifold
structure into the space of ``nearby'' $\nu$-partitions. The previous
section suggests a simple way of doing so. Namely, let $\{C_k, B_l\}$
be the smooth connected surfaces constituting the boundaries between
the sub-domains $P_j$ inside $\Omega$ ($\partial \Omega$ also
contributes to the boundaries of some of the sub-domains, but is not
taken into account, since it is not going to be changed). Let us fix a
smooth non-tangential to $\{C_k, B_l\}$ vector field $M$ in a
neighborhood $U$ of $\left(\cup_k C_k\right) \bigcup
\left(\cup_{l}B_l\right)$, which satisfies the conditions
imposed in the previous section (e.g., one can assume that outside of
a neighborhood of $\partial\Omega$, this is just the unit normal field
$N(x)$, which is smoothly modified near $\partial\Omega$ to be
tangential to $\partial\Omega$). Such a field always exists under the
genericity condition imposed on the partition $P$.

Let us pick a sufficiently large positive number $s$ (e.g., $s>(d+4)/2$ will suffice) and consider the space
\begin{equation}\label{E:space}
    \cF:=\left(\mathop{\oplus}\limits_k H^s(C_k)\right)\bigoplus \left(\mathop{\oplus}\limits_l H^s(B_l)\right),
\end{equation}
where $H^s$ is the standard Sobolev space of order $s$. We also introduce a continuous linear extension operator
\begin{equation}\label{E:extension}
    E: \cF\to H^{s+1/2}(U).
\end{equation}
In other words, the restriction of $E (f)$ to $C_k$ coincides with $f$. Such an extension operator is well
known to exist (see, e.g. \cite{LioMag_nonhomog72}). We will also assume that all the extended functions  $E (f)$ vanish outside of a small neighborhood of the nodal set, which can be achieved by multiplication by an appropriate smooth cut-off function.

Consider the ball $\cB_\rho$ of a small radius $\rho>0$ around the origin in the space $\cF$. Let $f\in \cB_\rho$ and $G_f$ be the shift by time $t=1$ along the trajectories of the vector field $E(f)(x)M(x)$. For a sufficiently small $\rho>0$, $G_f$ is a diffeomorphism, which preserves the boundary $\partial\Omega$. Its action on the surfaces $C_k$ and $B_l$ leads to another $\nu$-partition $G_fP$ of $\Omega$, which is close to the original partition $P$. We will denote this set of partitions by $\cP_\rho$ and identify it with the ball $\cB_\rho$ in the Hilbert space $\cF$. This, in particular, introduces the structure of a Hilbert manifold on $\cP_\rho$. Notice also that we will use consistent numbering of the sub-domains of partitions in $\cP_\rho$. Namely,
$$
(G_fP)_j=G_f(P_j),
$$
where, as before, for a partition $\Pi$ we denote by $\Pi_j$ its $j$th sub-domain, and $P_j$ are the sub-domains of the original (unperturbed) partition $P$.

We introduce now a mapping $\Xi$ from $\cP_\rho$ into $\R^{\nu}$ as follows: for a partition $\Pi$ we define
\begin{equation}\label{E:multimapping}
    \Xi (\Pi):=(\lambda_1(\Pi_1), \lambda_1(\Pi_2),\dots,\lambda_1(\Pi_\nu)).
\end{equation}
Let
$$
\Delta:=(\lambda,\lambda,\dots,\lambda)
$$
be the diagonal in $\R^\nu$.

\begin{definition}\label{D:equipartition}
  The set $\cE_\rho$ consists of all equipartitions in $\cP_\rho$, i.e.\
  $\Pi\in\cP_\rho$ satisfying
  $$
  \lambda_1(\Pi_1)=\lambda_1(\Pi_2)=\dots=\lambda_1(\Pi_\nu).
  $$
  To put it differently,
  \begin{equation}\label{E:inverse}
    \cE_\rho=\Xi^{-1}\left(\Delta\right).
  \end{equation}
\end{definition}

One notices that the restriction of the mapping $\Xi$ to $\cE_\rho$ is essentially the functional $\Lambda$ of the Introduction. More precisely,
$$
\Xi (\Pi)=(\Lambda(\Pi),\dots,\Lambda(\Pi)).
$$

We are ready now to prove the Proposition \ref{P:manifold}.

\subsection{Proof of the Proposition \ref{P:manifold}}

We prove first the following auxiliary result, which immediately implies the first statement of the Proposition \ref{P:manifold}.

\begin{lemma}\label{L:smooth}
For a small $\rho>0$, the mapping $\Xi:\cP_\rho\to \R^\nu$ is $C^\infty$.
\end{lemma}
\begin{proof}
It is sufficient to prove that for any $j$, the mapping
$$
f\to \lambda_1(G_f(P_j))
$$
is smooth as a mapping from the ball $\cB_\rho$ to $\R$, if $\rho$ is sufficiently small. Thus, one can restrict attention to a single sub-domain $D$. The often employed in such circumstances idea is to replace domain dependence with varying the coefficients of the differential operator in a fixed domain. Then the smooth dependence of $\lambda_1$ becomes a standard perturbation theory result (e.g., \cite{Gar_Sch_jam53,Hil_12}).

Let us consider first the case of a sub-domain $D$ that does not touch the boundary. Consider the mapping
$$
\Phi_f: x\to y:=x+E(f)(x)M(x)
$$
of the domain $D$ into $\Omega$. For sufficiently small $\rho>0$, it is a $C^2$-diffeomorphism of $D$ onto a ``nearby'' sub-domain $D^*$ of $\Omega$. The quadratic form of the operator $H(D^*)$ is given as
$$
\int\limits_{D^*}\left(\left|\frac{\partial u}{\partial y}\right|^2+V(y)|u(y)|^2\right)dy.
$$
Changing variables back from $y$ to $x$, we arrive at an operator
$H_f(D)$ in the fixed spatial domain $D$, but with variable
coefficients now:
$$
(H_f u)(x)=\nabla\cdot A_f(x)\nabla u(x) +V_f(x)u(x),
$$
where the matrix-valued function $A_f(x)$ is of class $C^1$, $f\to
A_f$ is a $C^\infty$-mapping from $H^s$ to the space of $C^1$-matrix
functions, and $f\to V_f$ is a $C^\infty$-mapping from $H^s$ to $C^2$.

We have thus replaced the domain dependence with the smooth dependence
of the coefficients of the operator. The operator $H_f$ acts
continuously from $H^2(\Omega)$ to $L^2(\Omega)$ and for $f=0$
coincides with $H(D)$. Moreover, the mapping $A_f\to H_f$ is a
continuous linear mapping from the space of $C^1$-matrix functions to
the space of bounded operators from $H^2(\Omega)$ to
$L^2(\Omega)$. Thus, for a sufficiently small $\rho$ we get a smooth
family of Fredholm operators between the aforementioned spaces. Due to
the simplicity of $\lambda_1(D)$, the standard perturbation theory
shows that $\lambda_1(H_f)$ depends smoothly on $f$, for a
sufficiently small radius $\rho$.

A similar consideration works when $D$ reaches the boundary of $\Omega$, i.e. at least one of the boundaries $B_l$ is involved. Without loss of generality, we can assume that only one such $B_l$ is involved. Introducing an appropriate smooth coordinate change, one can reduce consideration to the cylinder $B\times(-\varepsilon,\varepsilon)$ for small $\varepsilon>0$, with $\partial B\times (-\varepsilon,\varepsilon)$ as the corresponding part of $\partial\Omega$. Then the same reduction to a fixed domain but varying operator as before is possible, which again implies smooth dependence of $\lambda_1$ on $f$.
\end{proof}

Let us now address the second statement of the Proposition, that
$\cE_\rho$ is a smooth sub-manifold of $\cP_\rho$ of co-dimension
$\nu-1$. We will employ for this purpose the formula (\ref{E:inverse})
in conjunction with the domain variation formula, equation
\eqref{E:Mhadamard}, and a transversality theorem.

According to Lemma \ref{L:smooth}, the mapping $\Xi: \cP_\rho\to \R^{\nu}$ is a smooth mapping of Banach manifolds. The pre-image of the diagonal $\Delta\subset\R^{\nu}$ coincides, according to (\ref{E:inverse}), with $\cE_\rho$. We would like to know whether this pre-image is a smooth sub-manifold, and of what co-dimension. This is exactly the question tackled by the transversality theorems. Namely, if we can show that the mapping $\Xi$ is \term{transversal} to the diagonal one-dimensional sub-manifold $\Delta$ of $\R^{\nu}$, this will prove that the pre-image of $\Delta$ is a smooth sub-manifold of co-dimension $\nu-1$.

\begin{definition} (e.g.,
  \cite{Lan_difman02,Abr_62,Bour_man71})\indent The mapping $\Xi:
  \cP_\rho\to \R^{\nu}$ is \textbf{transversal} to $\Delta$, if at any
  point $\zeta = \Xi(v)$, $v\in\cP_\rho$, that belongs to $\Delta$, the
  vector sum of the tangent space $T_\zeta \Delta$ to $\Delta$ at
  $\zeta$ and of the range $D\Xi(T_v\cP_\rho)$ of the differential
  $D\Xi$ on the tangent space $T_v\cP_\rho$ is the whole space
  $\R^{\nu}$.
\end{definition}

\begin{theorem}(e.g., \cite[Sect. 3, Theorem 2]{Abr_62} or
  \cite[Sect. 5.11.7]{Bour_man71})\\
  If $\Xi$ is transversal to $\Delta$, then
  $\cE_\rho=\Xi^{-1}(\Delta)$ is a smooth sub-manifold of $\cP_\rho$
  of co-dimension $\nu-1$.
\end{theorem}

Thus, to finish the proof of the second statement of Proposition
\ref{P:manifold}, it only remains to prove the transversality of the
mapping $\Xi$ to the diagonal $\Delta$. The Rayleigh-Hadamard domain
variation formulas are helpful here.

The tangent space to the diagonal is spanned by the vector $(1, 1,
\ldots, 1)$.  We will demonstrate that the range of $D\Xi$ contains a
subspace of dimension at least $\nu-1$.  If the dimension is $\nu$,
there is nothing further to prove, but if the dimension is $\nu-1$ we
will show that the vector $(1, 1, \ldots, 1)$ \emph{does not} belong
to the subspace.

Consider the partition graph $\Gamma$ that corresponds to an
(automatically generic) $\nu$-partition $P\in\cP_\rho$. The vertices
of the graph correspond to the sub-domains $P_j$ and the edges to the
interfaces $C_j, B_l$. We will identify the target space $\R^\nu$ of
the mapping $\Xi$ with the space of real valued functions on the set
$V$ of vertices of the graph $\Gamma$. Consider a pair of adjacent
sub-domains $P_i$ and $P_j$ with the common part of their boundary $S$
(one of $C_j, B_l$). We consider the corresponding vertices $v_i, v_j$
and the edge $s$ of $\Gamma$.  We restrict our attention to functions
$f\in\cF$ that are non-zero on $S$ only and find the corresponding
directional (G\^{a}teaux) derivative of $\Xi$ at $P$ in the direction
of $f$.  Equation (\ref{E:Mhadamard}) shows that the only non-zero components of this
derivative correspond to vertices $v_i$ and $v_j$.
Since $M(x)$ is non-tangential to $S$, $M(x)\cdot N(x)$ is
sign-definite.  Choosing $f$ of the same sign we get
$$
D_f \Xi_i = -\int_S \left(\frac{\partial \psi_1(P_i)}{\partial
    n}\right)^2 f(x) M(x)\cdot N(x)dx \ < \  0
$$
and, similarly, $D_f \Xi_j > 0$.
Here we assumed that the normal $N(x)$ is directed outward with respect
to $P_i$ and, correspondingly, inward with respect to $P_j$.  We also
used the fact that, due to the standard uniqueness theorems, ${\partial \psi_1(P_i)}/{\partial
  n}$ is not everywhere zero.

Repeating this procedure for every pair of adjacent sub-domains we
arrive at a collection of vectors, one for each edge of the partition
graph, that have two non-zero components of the opposite sign each.
To characterize the space spanned by these vectors, we arrange them as
rows of a matrix and find its kernel.  Due to the connectedness of the
graph, the kernel is at most one-dimensional.  If the kernel is empty,
the vectors we found span all of $\R^\nu$.  If the kernel is spanned
by a vector $u$ (this is the case if the domain $\Omega$ is simply
connected, as this implies that the partition graph is a tree), then
$u$ must have all components of the same sign.  The vector $u$ is
therefore not orthogonal to the vector $(1,1,\ldots, 1)$, and the
latter vector complements the derivative vectors to span $\R^\nu$.

This finishes the proof of transversality and thus of the second
statement of the proposition. Since the first two claims of the
proposition imply the third one, the proof is completed.

\section{Proof of Theorem \ref{T:critical}}\label{S:critical}

\subsection{Proof of the equivalence
  $(\ref{I:critical:nodal})\Leftrightarrow
  (\ref{I:critical:Lambda_c})$.}

If $P$ is the nodal partition of a real-valued eigenfunction $\psi$, then, as we
have already mentioned before (Proposition \ref{P:ground_states}), the
restrictions of $\psi$ to the nodal domains are proportional to the
groundstates in these domains.  Denote these proportionality constants
by $a_j$.  Since the eigenfunction $\psi$ is continuously
differentiable, the groundstates $\psi_1(P_j)$ scaled with the
corresponding factors $a_j$ have matching normal derivatives at the
common boundaries:
\begin{equation}
  \label{E:derivatives_match}
  a_i\left.\frac{\partial \psi_1(P_i)}{\partial n} \right|_{\ S}
  = a_j\left.\frac{\partial \psi_1(P_j)}{\partial n}\right|_{\ S},
\end{equation}
where sub-domains $P_i$ and $P_j$ have the common boundary $S$, and
$n$ is a normal vector to $S$.

Now let $c_k=a_k^2$ and consider the G\^{a}teaux derivative of the
functional $\Lambda_c$ in the direction $f$ that is non-zero only on a
single boundary between the sub-domains $P_i$ and $P_j$.  Since the
only affected terms in $\Lambda_c$ are $a_i^2\lambda(P_i)$ and
$a_j^2\lambda(P_j)$, the derivative is
\begin{equation}\label{E:difference}
  \int_S \left(a_j^2 \left(\frac{\partial \psi_1(P_j)}{\partial
      n}\right)^2 -a_i^2 \left(\frac{\partial \psi_1(P_i)}{\partial
      n}\right)^2\right)f(x) M(x)\cdot N(x)dS,
\end{equation}
where we applied the Rayleigh-Hadamard formula, equation~\eqref{E:Mhadamard}.  The
difference in signs arises since $N(x)$ points outward with respect to
$P_i$ but inward with respect to $P_j$.  Now we observe that, due
to (\ref{E:derivatives_match}), the integrand is identically equal to zero and thus the G\^{a}teaux derivative in the direction of $f$ is equal to zero. The same is obviously true for arbitrary variations $f$, involving any number of boundaries. Thus, the nodal partition is a critical
point of the functional $\Lambda_c$.

Conversely, if a partition $P_\rho$ is a critical point of
$\Lambda_c$, we get that the G\^{a}teaux derivative of $\Lambda_c$ is zero in any direction $f(x)$.  This implies the equality
\begin{equation*}
  c_i\left(\frac{\partial \psi_1(P_i)}{\partial n}\right)^2
  = c_j\left(\frac{\partial \psi_1(P_j)}{\partial n}\right)^2
\end{equation*}
on the common boundary $S$ of any two neighboring domains $P_i$ and
$P_j$.  Setting $\alpha_k = \pm \sqrt{c_k}$ and choosing the signs so
that any two neighboring domains have different signs (possible due to
bipartiteness) ensures that (\ref{E:derivatives_match}) is satisfied.
Then the function $\psi$ defined by
\begin{equation*}
  \psi\mid_{P_k} = a_k \psi_1(P_k)
\end{equation*}
is an eigenstate of $H$.

\subsection{Proof of the equivalence $(\ref{I:critical:Lambda_c})
  \Leftrightarrow (\ref{I:ctitical:on_equipartns})$}

If $P$ is a critical point on $\cP_\rho$ of the functional
$\Lambda_c$, then the restriction of $\Lambda_c$ to $\cE_\rho$ is a
critical point on $\cE_\rho$.  But on $\cE_\rho$ any functional
$\Lambda_c$ coincides with $\Lambda$.

Conversely, assume that $P$ is a critical point of $\Lambda$ on
$\cE_\rho$. We can extend the functional $\Lambda$ to the whole
$\cP_\rho$ as $\lambda_1(P_1)$. Since $\cE_\rho$ can be given by the
smooth relations $\lambda_1(P_1)-\lambda_1(P_j)=0$ for $j=2,3,\dots
,\nu$, the Lagrange multiplier method implies that $P$ must be a
critical point of a non-trivial linear combination $\Lambda_b:=\sum
b_j\lambda_1(P_j)$.  All $b_j$ are of the same sign: otherwise there
are two neighboring domains with $b_j$ of different signs and the
variation of $\Lambda_b$ with respect to the boundary between the two
domains cannot be zero for a sign-definite $f(x)$ (see equation
(\ref{E:difference})). Thus the vector of coefficients $b_j$ can be
normalized to be a unit simplex vector.  This finishes the proof of
Theorem \ref{T:critical}.

\section{Proof of Theorem \ref{T:morse}}\label{S:morse}

Let us present first the strategy of the proof. As we have already
seen, it is sometimes useful to play with different extensions of the
functional $\Lambda$ from the (local patch of the) space $\cE_\rho$ of
equipartitions to a larger manifold. While previously it was the
(local patch of the) space of all partitions, now we need some further
extension. Indeed, we would like to compare somehow the nodal count
$\nu$ (which is fixed) with the consecutive number $n$ of an
eigenfunction $\psi_n$. It is hard to observe where the information
about $n$ is hidden in the spaces of partitions themselves. On the
other hand, the quadratic form (or Rayleigh quotient) contains this
information.  To use it we will extend the functional $\Lambda$ to a
larger space that is a functional space, not just a set of domain
partitions. Then we will have to restrict back in order to compare the
Morse indices of $\Lambda$ and of its extension.

Before implementing this program, in the following sub-sections we
start proving some auxiliary statements that will come handy later on.

\subsection{Critical points on direct sums of spaces}

\begin{theorem}
  \label{T: reduction}
  Let $X=Y \bigoplus Y'$ be a direct decomposition of a Banach space.  Let also
  $f:X\to \R$ be a smooth functional such that $(0,0)\in X$ is its
  critical point of Morse index $m$.

  If for any $y$ in a neighborhood of zero in $Y$, the point
  $(y,0)$ is a critical point of $f$ over the affine subspace
  $\{y\}\times Y'$, then the Hessian $F_2$ of $f$ at the origin, as
  a quadratic form in $X$, is reduced by the decomposition $X=Y\bigoplus
  Y'$.

  In particular, let $Y$ be the locus of minima of $f$ over the affine
  subspaces $\{y\}\times Y'$, i.e.
  \begin{equation}
    \label{eq:locus_min}
    (y,0) = \arg \min_{y' \in Y'} f(y,y'),
  \end{equation}
  for any $y$ in a neighborhood of zero in $Y$.  Then the Morse index
  of $0$ as a critical point of the restriction $f|_{Y}$ is equal to
  $m$ (i.e., the same as the Morse index of this point on the whole
  neighborhood of zero in $X$).

  Finally, if $(0,0)$ is a non-degenerate critical point of $f$ on
  $X$, then $0$ is non-degenerate as a critical point of $f|_{Y}$.
\end{theorem}

\begin{proof}
  We can assume, without loss of generality, that $f(0,0)=0$. Using
  this and the condition of the criticality of the origin, the Taylor
  formula of the second order for $f$ on $X$ near the origin is
  $$
  f(y,y') = A(y,y) + B(y',y') + C(y,y')  + \mbox{ higher order terms}.
  $$
  Here $A$ is a quadratic form in $Y$, $B$ is a
  quadratic form in $Y'$ and $C$ is a bilinear form acting on $Y\times
  Y'$.  We will now take the G\^ateaux derivative of $f(y,y')$ in the
  direction $z' \in Y'$ and evaluate it at $(y, 0)$.

  The derivative of $A(y,y)$ is zero (since $A$ does not depend on
  $y'$).  The derivative of $B(y',y')$ is $B(z', y') + B(y', z')$,
  which vanishes after substituting $y'=0$.  Thus we find
  \begin{equation*}
    D_{(0,z')}f(y,0) = C(y,z'),
  \end{equation*}
  which must be zero since $(y,0)$ is a critical point over
  $\{y\}\times Y'$.  Since $z'$ and $y$ are arbitrary (provided $y$ is
  sufficiently close to 0), the bilinear form $C$ is identically
  zero.  The Hessian is thus
  \begin{equation*}
    F_2 = A(y,y) + B(y',y').
  \end{equation*}
  This proves the first statement of the Theorem.

  If condition~\eqref{eq:locus_min} is fulfilled, then $B(y',y')$ is
  non-negative definite and the negative subspaces of $F_2$ and $A$
  coincide (after the natural projection).

  Finally, if $(0,0)$ is non-degenerate, it implies that $\mu^0(A) =
  \mu(A)$, where $\mu$ denotes the Morse index.  Indeed, otherwise
  $\mu^0(F_2) \geq \mu^0(A) > \mu(A) = \mu(F_2)$ which contradicts
  the non-degeneracy of $(0,0)$.
\end{proof}

\begin{remark}
  It is obvious from the proof of the theorem that its second
  statement can be generalized in the following manner: if $Y$ is the
  locus of \emph{critical points} of $f$ of index $m'$ (locally
  independent of $y$) then the Morse index of the restriction $f|_{Y}$
  is equal to $m-m'$.  We will only need $m'=0$ in the present
  manuscript but the more general version becomes useful in other
  contexts, such as \cite{Ber_prep11}.
\end{remark}

\subsection{Some objects needed for the proof}

Due to the local character of Theorem~\ref{T:morse}, all the
constructions below are needed only locally, near a generic
eigenfunction $\psi_n$ indicated in the statement of the theorem, and
correspondingly near its nodal partition.  The basic notions and facts
concerning finite- or infinite- dimensional vector-bundles that we use
below can be found in many standard sources on topology (e.g.,
\cite{Atiyah,Husemoller}) or in the survey \cite{ZaiKreKucPan_75},
where such bundles are studied in relation to the operator theory.

We will be considering again the (local) manifold $\cP_\rho$ of
partitions ``close'' to the nodal partition $P(\psi_n)$ and its
sub-manifold $\cE_\rho$ of codimension $\nu-1$ that consists of
equipartitions only.

\begin{definition}
  \label{def:bundle_B}
  In the trivial bundle $\cP_\rho\times H^1_0(\Omega)$ over
  $\cP_\rho$, we consider a fibered subset $B$ that has the fiber over
  a partition $P$ consisting only of functions vanishing on the
  partition's interfaces $N \cup \partial \Omega$.  In other words,
  this fiber is $H^1_P:=\bigoplus_j H^1_0(P_j)$.
\end{definition}

\begin{lemma} \label{L:bundleB} $B$ is a smooth locally trivial
  vector sub-bundle of the trivial bundle
  $$\cP_\rho\times H^1_0(\Omega)\mapsto \cP_\rho.$$
\end{lemma}

\begin{proof}
  The proof follows the same line as the one of Lemma
  \ref{L:smooth}. Namely, the dependence on the partition $P$ is
  replaced, using a smooth family of diffeomorphisms, with a fixed
  partition $\Pi$.  The corresponding change of variables in the
  functions from $H^1_P$ maps them to $H^1_\Pi$.  Then $B$ becomes
  just the trivial sub-bundle $\cP_\rho \times \left(\bigoplus_j
    H^1_0(\Pi_j)\right)$ in $\cP_\rho\times H^1_0(\Omega)\mapsto
  \cP_\rho$. Inverting the diffeomorphisms provides a smooth
  trivialization of $B$, which proves the lemma.
\end{proof}

\begin{definition}We denote by $SB$ the locally-trivial bundle of the
  unit (in $L_2$-norm) spheres of the fibers of $B$.

  The restrictions of $B$ and $SB$ to $\cE_\rho$ (clearly
  locally-trivial bundles) will be denoted by $B_E$ and $SB_E$
  correspondingly.
\end{definition}

We will now restrict the bundle $B$ further.

\begin{definition}We denote by $C$ the vector bundle whose fiber over a partition $P$ consists of functions of the form $\sum_j c_j\psi_1(P_j)$, where $c_j$ are real constants and $\psi_1(P_j)$ is the normalized positive groundstate on the sub-domain $P_j$.

Correspondingly, $SC$ consists of the unit (in $L_2$-norm) spheres of $C$ and $C_E$ and $SC_E$ are restrictions of the corresponding fibered sets to $\cE_\rho$.
\end{definition}

The following lemma shows that $C$ and $SC$ are locally-trivial
bundles, and, in particular, $C$ is a locally trivial vector bundle.

\begin{lemma} \label{L:bundleC}
$C$ is a smooth $\nu$-dimensional vector sub-bundle of $B$, and thus of the trivial bundle
$$\cP_\rho\times H^1_0(\Omega)\mapsto \cP_\rho.$$
\end{lemma}
\begin{proof}
  The proof follows the same line as in Lemmas \ref{L:smooth} and
  \ref{L:bundleB}. Namely, after applying a smooth family of
  diffeomorphisms, one deals with a fixed partition $\Pi$, but instead
  with the operator whose coefficients depend smoothly on $P$.
  Perturbation theory shows that the corresponding ground-state $f_j$
  in each sub-domain $\Pi_j$ depends smoothly on $P$, as a vector in
  $H^1_0(\Pi_j)$. We extend it, without changing the notation, by zero
  to the whole domain $\Omega$. Then $\{f_j\}$ is a smoothly dependent
  on $P$ frame of $\nu$ linearly independent vectors in
  $H^1_0(\Omega)$. Thus, this frame spans a smooth finite-dimensional
  vector-bundle. After applying the inverses of the diffeomorphisms,
  we get the claim of the lemma.
\end{proof}

\subsection{The functional $\Lambda$ and the quadratic form of $H$}

Consider the quadratic form on $H^1_0(\Omega)$
$$
Q[f]:=\int_\Omega \left( |\nabla f(x)|^2 + V(x) |f(x)|^2 \right) dx.
$$
It can, by restricting to each fiber, be defined as a smooth functional on the vector bundle $B$ and its sub-bundles that we considered above.

\begin{lemma}\label{L:indn-1}
  The point $(P(\psi_n),\psi_n)$ in $SB$ is a non-degenerate critical
  point of $Q$ of Morse index $\mu=n-1$.
\end{lemma}

\begin{proof}It is clear that $Q$ on $H^1_0(\Omega)$ has an $n-1$-dimensional subspace on which its Hessian at $\psi_n$ is negative. Namely, this is the subspace generated by the eigenfunctions $\psi_1,...,\psi_{n-1}$. If we show that these directions are among the tangential ones to $SB$ at $(P(\psi_n),\psi_n)$, this will prove that $\mu\geq n-1$.

  Due to the locally-trivial structure of $SB$, there are two main
  ways to get tangential vectors to $SB$.  One is to vary the partition
  $P$ (which will give ``horizontal'' tangent vectors to the bundle
  $SB$).  The other is to keep the partition fixed, while varying
  the function $\psi_n$, keeping the nodal set fixed (``vertical''
  tangent vectors). The vertical tangent vectors are just arbitrary
  functions in $H^1_0(\Omega)$ that vanish on the nodal set $\cZ$ of
  the partition $P(\psi_n)$ (we have previously denoted this space
  $H^1_P$). The horizontal tangential vectors look as follows:
$$
E(f) M\psi_n=E(f) N\psi_n(x)(M(x)\cdot N(x)),
$$
where $f\in H^s$ on the nodal set of $\psi_n$ is the function defining this set's infinitesimal variation. Also, the notation $Xg$ for a vector field $X$ and a function $g$ (e.g., $M\psi_n$ and $N\psi_n$) mean the derivative of the function $g$ along the field $X$. Notice that, since $M(x)$ is non-tangential to $\cZ$, $M(x)\cdot N(x)$ is a smooth separated from zero function on the nodal set.

We now show that any $\psi_j$ with $j<n$ can be represented as a sum of vertical and horizontal vectors. We notice that the genericity condition requires in particular that zero is a regular value of the normal derivative of $\psi_n$ on $\partial\Omega$. This implies that the derivative $N \psi_n$ is a smooth function that has a non-degenerate zero at $\partial\Omega$ on the nodal set $\cZ$. Since $\psi_j$ is a smooth function vanishing on $\partial\Omega$, the function
$$
f(x):=\frac{\psi_j(x)}{(N \psi_n) M(x)\cdot N(x)}
$$
belongs to $H^s$ on the nodal set. Hence, the horizontal tangent vector in the direction of $f$
$$
h:=(N \psi_n)E(f)M(x)\cdot N(x)
$$
coincides with $\psi_j$ on the nodal set. This means that the difference $g:=\psi_j-h$ belongs to $H^1_P$, and thus is a vertical tangent vector. This shows that each eigenfunction $\psi_1,...,\psi_{n-1}$ can be represented as the sum of a vertical and horizontal vectors and thus is tangent to $SB$. This proves the estimate $\mu\geq n-1$ for $SB$.

We will now prove that the index $\mu^0$ cannot exceed $n-1$.  This
will verify that $\mu=n-1$ and that the critical point is
non-degenerate.

Suppose that there is an $n$-dimensional subspace $L$ in the tangent
space to $SB$ at $(P(\psi_n),\psi_n)$, where the Hessian of $Q$ is
non-positive. Since each fiber of $SB$ consists of functions from the
space $H^1_0(\Omega)$, there is a tautological mapping $(P,f)\mapsto
f$ from $SB$ into $H^1_0(\Omega)$ (in fact, into the set of functions
of unit $L_2$-norm). If the Fr\'{e}chet derivative of the tautological
mapping has zero kernel, then the subspace $L$ will produce an
$n$-dimensional subspace of $H^1_0(\Omega)$ transversal to $\psi_n$,
where the Hessian of $Q$ at $\psi_n$ is non-positive, which is a
contradiction. So, let us show that the kernel of the Fr\'{e}chet
derivative is zero. Due to the local trivial structure of $SB$, one
sees that the image of any tangent vector under the Fr\'{e}chet
derivative has the form
$$
g=E(f)(x)\left(M\psi_n\right)+v(x),
$$
where the function $f\in\cF$ is responsible for the infinitesimal
variation of the nodal set of $\psi_n$, $E$ is the previously
introduced extension operator from the nodal set to $\Omega$, and a
function $v\in H^1_P(\Omega)$ corresponds to the infinitesimal
variation in the fiber direction (i.e., the pair $(f,v)$ describes a
tangent vector to $SB$). Suppose now that $g=0$. In particular,
$g|_{\cZ}=0$. Taking into the account that $v|_{\cZ}=0$ (which is true
for any function from the space $H^1_P$), we conclude that
$$
E(f)(x)\left(M\psi_n\right)|_{\cZ}=f(x)M(x)\cdot\nabla \psi_n(x)|_{\cZ}=0.
$$
Since $\nabla \psi_n(x)|_{\cZ}$ is nonzero and orthogonal to $\cZ$, and
$M$ is transversal to $\cZ$, we see that
$$
M(x)\cdot\nabla \psi_n(x)|_{\cZ}\neq 0.
$$
This implies that $f(x)$ (defined on $\cZ$ only) vanishes
identically. Due to the linearity of the extension operator, the
extension $E(f)$ vanishes everywhere in $\Omega$. Hence, $g=v$. Since
$g=0$, we conclude that $v=0$ and thus the Fr\'{e}chet derivative of
the tautological mapping is injective.  This finishes the proof of the
lemma.
\end{proof}

\begin{lemma}
  \label{L:SC}
  The point $(P(\psi_n),\psi_n)$ in $SC$ is a non-degenerate critical
  point of $Q$ of Morse index $n-1$.
\end{lemma}

\begin{proof}
  We know that the eigenfunction $\psi_n (x)$ has the form
  $$
  \psi_n (x)=\sum_{j=1}^\nu c_j\psi_1(D_j)(x),
  $$
  where $\{D_j\}$ is the nodal partition of $\psi_n$ and each ground
  state $\psi_1(D_j)$ is extended by zero to the whole domain
  $\Omega$, $\sum c_j^2=1$, and none of the coefficients $c_j$ are
  equal to zero.  Thus the coefficients $c_j$ remain bounded away from
  zero in a neighborhood $U$ of $\psi_n$ in $SC$.

  As we have already seen, $C$ is a smooth $\nu$-dimensional
  sub-bundle of $B$ and thus $SC$ is a smooth sub-manifold of $SB$. We
  will now introduce near the point $(P(\psi_n),\psi_n)$ a smooth
  foliation of $SB$ by manifolds transversal to $SC$, such that $SC$
  will be the locus of minima of $Q$ over the leaves of this
  foliation. Then, after a smooth local change of coordinates, we will
  be in the situation of Theorem \ref{T: reduction}. Thus, the Morse
  index of $Q$ at $(P(\psi_n),\psi_n)$ on $SC$ will be equal to the
  one on $SB$, which will prove the first claim of the lemma.

  So, let us finish the proof by constructing such a
  foliation. Consider the following mapping $\Upsilon$ from a
  neighborhood $U$ of $(P(\psi_n),\psi_n)\in SB$ to $\cP_\rho\times
  S^{\nu-1}$, where $S^{\nu-1}$ is the unit sphere in $\R^\nu$:
  $$
  \Upsilon(P,f)=(P, \|f\|_{L_2(P_1)},\dots,\|f\|_{L_2(P_\nu)}).
  $$
  Notice that none of the components $\|f\|_{L_2(P_j)}$ of the vector
  $\Upsilon(P,f)$ vanishes (since this is the case for
  $\Upsilon(P(\psi_n),\psi_n)$).

  By the arguments provided before, this is a smooth mapping. It is
  also clear that it is a submersion.

  Let $w=(P, \sum c_j \psi_1(P_j))$ be a point in $U\cap SC$ near
  $(P(\psi_n),\psi_n)$. Consider the leaf
  $L_w=\Upsilon^{-1}(P,c_1,c_2,\dots,c_\nu)$. Due to the submersion
  property of $\Upsilon$, the leaves $L_w$ form near
  $(P(\psi_n),\psi_n)$ a smooth fibration. Since the differential of
  $\Upsilon$ on $SC$ is surjective, this foliation is transversal to
  $SC$.  Moreover, the groundstate $\psi_1(P_j)$ is the location of the
  minimum of the corresponding Rayleigh quotient, and therefore
  \begin{equation*}
    Q[f] = \sum_{j=1}^\nu Q\left[f|_{P_j}\right]
    \geq \sum_{j=1}^\nu Q\left[ c_j \psi(P_j)\right],
    \qquad c_j = \|f\|_{L_2(P_j)}.
  \end{equation*}
  That is, the minimal value of $Q$ on the leaf $L_w$ is attained
  exactly when $f_j=c_j\psi_1(P_j)$ for all $j$, i.e. on $SC$.
\end{proof}

\begin{lemma}\label{L:pullback}
  The point $(P(\psi_n),\psi_n)$ is a (degenerate) critical point of
  $Q$ in $SC_E$ with the Morse index not less than $n-\nu$ and the
  $\mu^0$-index not more than $n-1$.

  The restriction of the quadratic form $Q$ to the fibers of $SC_E$ is
  the pull-back of the functional $\Lambda$ from the base $\cE_\rho$.
  In other words, if $\pi:SC_E \mapsto \cE_\rho$ is the bundle
  projection, then for any $x\in SC_E$,
  $$
  Q(x)=\Lambda(\pi(x)).
  $$
\end{lemma}

\begin{proof}
  According to Proposition \ref{P:manifold}, $SC_E$ is a smooth
  sub-manifold of $SC$ of co-dimension $\nu-1$. Hence, upon
  restriction to $SC_E$, an index (either Morse or $\mu^0$) cannot
  decrease by more than $\nu-1$.  A restriction also can not increase
  an index.  These observations, together with Lemma~\ref{L:SC}, prove
  the first two claims.

  Let $f=\sum_jc_j\psi_1(P_j)$ be an element of the fiber of $SC_E$
  over a partition $P=\{P_j\}$. Then
  $$
  Q(f)=\sum |c_j|^2\lambda_1(P_j).
  $$
  Since $P$ is an equipartition, all the values $\lambda_1(P_j)$ are
  equal to the same value $\Lambda(P)$. Taking into account that
  $\sum|c_j|^2=\|f\|^2_{L_2(\Omega)}=1$, we get
  $$
  Q(f)=\sum |c_j|^2\lambda_1(P_j)=\Lambda(P)\left(\sum|c_j|^2\right)=\Lambda(P),
  $$
  which proves the last statement of the lemma.
\end{proof}

\subsection{Index of $\Lambda$ and nodal deficiency}

We are now ready to prove Theorem \ref{T:morse} by obtaining
estimates for the Morse index $\mu$ from two sides.

\begin{proof}[Proof of the estimate from below: $\mu\geq n-\nu_{\psi_n}$]
  From Lemma~\ref{L:pullback} we see that the pull-back of $\Lambda$
  to $SC_E$ has Morse index at least $n-\nu$ at the point
  $(P(\psi_n),\psi_n)$, i.e. there is a tangent subspace of dimension
  at least $(n-\nu)$, where the Hessian of $Q$ is strictly negative
  definite.  On the other hand, the pullback is constant along the
  fibers of $\pi:SC_E \mapsto \cE_\rho$.  Therefore the above subspace
  must correspond to such a subspace for the Hessian of $\Lambda$.
\end{proof}

\begin{proof}[Proof of the estimate from above: $\mu\leq \mu^0 \leq
  n-\nu_{\psi_n}$]
  We use Lemma \ref{L:pullback} again.  It shows that the index
  $\mu^0$ of $\Lambda$ at $P(\psi_n)$ cannot exceed
  $(n-1)-(\nu-1)=n-\nu$.  Indeed, the pullback of $\Lambda$ to $SC_E$
  add a $(\nu-1)$-dimensional subspace on which the Hessian is zero.
  But the $\mu^0$-index on $SC_E$ is bounded by $n-1$.
\end{proof}

To summarize, together with the trivial inequality $\mu\leq \mu^0$
(see the discussion prior to Theorem~\ref{T:morse}) we have shown that
\begin{equation*}
  n-\nu_{\psi_n} \leq \mu \leq \mu^0 \leq n-\nu_{\psi_n}.
\end{equation*}
This shows that the critical point is non-degenerate
(i.e. $\mu=\mu^0$) and of Morse index $\mu = n-\nu_{\psi_n}$,
finishing the proof of Theorem \ref{T:morse}.

\section{Final remarks and conclusions}\label{S:remarks}

\begin{enumerate}
\item The results of this paper (Theorems \ref{T:critical} and
  \ref{T:morse}) translate without any changes in their proofs to the
  case when $\Omega$ is a compact smooth Riemannian manifold with or
  without boundary.
\item Smoothness conditions imposed on the domain, potential, and
  partition interfaces, can certainly be weakened. In this text, we
  have decided not to do so, in order not to complicate considerations
  unnecessarily.
\item The sets $\cP$ and $\cE$ that we considered involved only
  generic partitions, which allows perturbing the boundaries of
  sub-domains independently, simplifying the structure of the manifold
  of partitions and the consequent considerations of the text. In
  general, however, the smooth pieces of partition manifolds are
  joined into singular ``varieties'' $\cP$ and $\cE$, where the
  junctions occur when the partition interfaces start meeting each
  other. It would be interesting to see whether one could prove an
  analog of Theorems \ref{T:critical} and \ref{T:morse} for such
  non-generic partitions. The authors believe that something of this
  nature can be done.
\item An interesting connection between the zeros of eigenfunctions
  and stability index with respect to another suitably defined
  perturbation has emerged recently.  In the preprint
  \cite{Ber_prep11} (see also \cite{CdV_prep12} for an alternative
  proof), it was shown for the graph case that the ``surplus'' number
  of zeros $\phi_n - (n-1)$ is equal to the Morse index of the
  eigenvalue considered as a function of magnetic perturbation of the
  Hamiltonian.  Here $\phi_n$ is the number of zeros of the $n$-th
  eigenfunction.
  \item The reader probably has noticed that there is a significant
  flexibility in our choice of the functional (e.g., Sobolev) spaces. Probably,
  one could achieve the same results working in the $C^k$ spaces of smooth
  functions. Our choice of the Sobolev scale was due to a known simpler
  interpretation of the Morse theory on Hilbert, rather than Banach, manifolds.
\end{enumerate}

\section*{Acknowledgments}\label{S:thanks}
The authors thank R.~Band, Y.~Colin~de~Verdiere, L.~Friedlander, B.~Helffer, H.~Hezari,
T.~Hoffmann-Ostenhof, H.~Raz, P.~Sarnak, and S.~Zelditch for useful
discussions and references. We also are grateful to the reviewer for many important remarks.
The work of G.B. was partially supported
by the National Science Foundation grant DMS-0907968.  The work of
P.K. was partially supported by MSRI and IAMCS.  The work of
U.S. was partially supported by The Wales Institute of Mathematical
and Computational Sciences (WIMCS), EPSRC (grant EP/G021287), and ISF
(grant 166/09).



\end{document}